\documentclass[conference]{IEEEtran}
\usepackage{amssymb}
\usepackage{amsmath}
\usepackage{graphicx}
\usepackage{subfigure}
\usepackage{color}

\input{tcilatex}

\title{Memory Allocation in Distributed Storage Networks}

\author{Mohsen Sardari$^\dag$, Ricardo Restrepo$^\ddag$, Faramarz
Fekri$^\dagger$, Emina Soljanin$^\ast$ \\ $^\dag$School of Electrical and
Computer Engineering, Georgia Institute of Technology, Atlanta, GA 30332\\
$^\ddag$School of Mathematics, Georgia Institute of
Technology, Atlanta, GA 30332\\
$^\ast$Alcatel-Lucent Bell-Labs, Murray Hill, NJ 07974 \\
\texttt{Email:}$^\dag$\{mohsen.sardari, fekri\}@ece.gatech.edu,
$^\ddag$restrepo@math.gatech.edu, $^\ast$emina@research.bell-labs.com}

\begin{document}
\maketitle

\begin{abstract}

We consider the problem of distributing a file in a network of storage nodes
whose storage budget is limited but at least equals the size file. We first
generate $T$ encoded symbols (from the file) which are then distributed among the
nodes. We investigate the optimal allocation of $T$ encoded packets to the
storage nodes such that the probability of reconstructing the file by using any
$r$ out of $n$ nodes is maximized. Since the optimal allocation of encoded
packets is difficult to find in general, we find another objective function which
well approximates the original problem and yet is easier to optimize. We find the
optimal symmetric allocation for all coding redundancy constraints using the
equivalent approximate problem. We also investigate the optimal allocation in
random graphs. Finally, we provide simulations to verify the theoretical results.
\end{abstract}

\section{Introduction}
A file in a distributed storage network can be replicated throughout the network
to improve the performance of retrieval process, measured by routing efficiency,
persistence of the file in the network when some storage locations go out of
service, and many other criteria. Most of the studies in network file storage
consider a common practice where every node in the network either stores the
entire file or none of it. In an important article, Naor and Roth~\cite{Naor1991}
studied how to store a file in a network such that every node can recover the
file by accessing only the portions of the file stored on itself and its
neighbors, with the objective of minimizing the total amount of data stored. By
applying MDS (Maximum Distance Separable) codes and generating codeword symbols
of the file, they presented a solution that is asymptotically optimal in
minimizing the total number of stored bits, when the original file has a length
much larger than the logarithm of the graph's degree of the storage network.
Other works~\cite{Jiang2003,Jiang2005} extended the result of~\cite{Naor1991} and
devised algorithms for memory allocation in tree networks with heterogeneous
clients. Distributed storage is also studied in sensor
networks~\cite{Zhenning2009,Dimakis2005}. In sensor networks, the focus is
usually on the data retrieval assuming that a data collector has access to a random
subset of storage nodes while in this paper we address the allocation problem.

One of the appealing features for a distributed storage system is the ability to
scale the persistence of data arbitrarily up and down on-demand. In other words,
the cost of accessing the stored data should be adjustable based on the demand.
In one extreme, all the nodes have ``easy'' access to the stored file, either by
storing the whole file or a large part of it. On the other extreme, just a single
node stores the file entirely and other nodes need to fetch the file from that
node. It is clear that by making more copies of a file and spreading those copies
in the network, the retrieval of the file becomes easier. The use of MDS codes
provides the flexibility to increase the persistence of a file gradually. For
example, for a given file of size $F$, we can generate $T$ symbols using a
$(T,F)$ MDS code such that every $F$-subset of those $T$ symbols is sufficient to
reconstruct the original file. We call $T$ the budget considered for the file.
Now, the question is as to how increasing the budget of a file affects the
retrieval process. In order to answer this question, we need to consider a model
for data retrieval. Recently, Leong et. al.~\cite{Leong2009a} investigated this
problem and introduced the following model for the network. Consider a network
with $n$ storage nodes. We distribute a file of size $F$ and budget $T$ (packets
or symbols) among these storage nodes. Then, we look at all the possible subsets
of size $r$ of the storage nodes. We say that a specific $r$-subset is successful
in recovering the file if the total number of packets stored in that subset of
the nodes is at least the file size $F$. We are to find the best assignment of
these $T$ symbols to $n$ storage nodes such that the maximum number of the
$r$-subsets of storage nodes have enough number of symbols to reconstruct the
file. The rational behind the model is that in a real storage network, every node
can be reached by all the other nodes in network. Once a retrieval request for a
file is received by a node in network, the node tries to fetch all the parts of
the file and respond to the request. The cost of fetching the parts from
different nodes is not equal (other nodes may be down, busy, etc.). Therefore, in
the model we assume that each node fetches the necessary parts of the file from
the other $r-1$ most accessible nodes.

In general, this problem is quite challenging and the optimal allocation is
non-trivial. In~\cite{Leong2009a}, the authors provide some results for the
symmetric allocation and probability-1 recovery regime which is a special case of
the problem introduced in~\cite{Naor1991}. Symmetric allocation refers to a
scheme where, based on the budget, we split the storage nodes into two groups:
the nodes with no stored symbols and the nodes that store the same number of
symbols. In probability-1 recovery regime, all the nodes should be able to
reconstruct the file. As illustrated in~\cite{Leong2009a}, the optimal allocation
is not obvious even if we only consider the symmetric allocations.

For very low budgets, we observe that the budget is concentrated over a minimal
subset of storage nodes in the optimal allocation. On the other hand, for high
budget levels, we observe a maximal spread of budget over storage nodes. It is of
interest to determine as to how this transition occurs and also to study the
behavior of the optimal allocation versus budget. In this paper, we take the
initial steps towards the characterization of the optimal allocation. In
section~\ref{sec:File Allocation Problem}, we give the formal definition of the
problem and the model we consider. Then, in Section~\ref{sec:proofresult} we
prove that an easier to solve problem well approximates the original problem.
Using the alternative approach, we solve the file allocation problem for
symmetric allocations (Section~\ref{sec:optimalsymmetric}); we also consider
symmetric allocations in random graphs. Finally, simulation results are provided
in Section~\ref{sec:simulation}.


\section{File Allocation Problem}
\label{sec:File Allocation Problem}

\subsection{Problem Statement}
\label{subsec:Problem Statement}

We are given a file of size $F$ and a network with budget $T$. We generate $T$
redundant symbols using a $(T,F)$ MDS code. An allocation of $T$ symbols to $n$
nodes is defined to be a partition of $T$ into $n$ sets of sizes $x_{1},\ldots
,x_{n}$, where $x_{i}$ is the number of symbols allocated to the $i$th storage
node. Note that $\sum x_{i} = T$ and $x_{i}\geq 0~$for $i=1,\ldots ,n$. Our goal
is to find an allocation which maximizes the number of $r$-subsets jointly
storing $F$ or more packets.

Combination networks provide a simple illustration of the allocation problem
under study. As shown in Figure~\ref{fig:comb_net}, \begin{figure}[htp]
\begin{center}
  \includegraphics[width=2.3in]{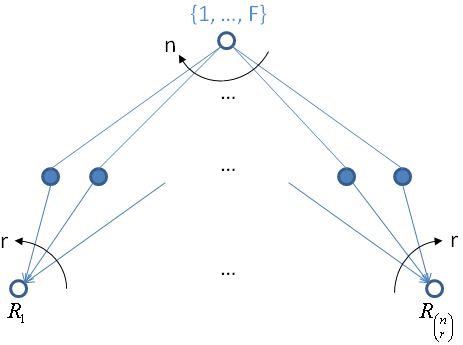}
  \vspace{-.1in}
  \caption{Combination Network. Virtual source node in layer one has a file of
  size $F$. The solid nodes in layer two represent the storage nodes in network.
  The third layer contains virtual receiver nodes. Each receiver node
  corresponds to an $r$-subset of the storage nodes.}
  \label{fig:comb_net}
  \vspace{-.15in}
\end{center}
\end{figure}
there are three layers of nodes. A virtual source node in layer 1 has a file
of size $F$ (packets) to be distributed among storage nodes in layer 2. There are
$n$ storage nodes in layer two which represent the actual storage nodes in the
storage network. Attributed to the file is a budget $T$. The data retrieval
phase is visualized in the third layer of the combination network, which contains
$\binom{n}{r}$ virtual receiver nodes. Each receiver node corresponds to an
$r$-subset of the storage nodes. We are going to find the best allocation of the
budget $T$ such that the maximum number of receivers $R$ in the combination
network can reconstruct the file. Indeed, the success of the recovery process depends on
the budget $T$. Based on the illustration in Figure~\ref{fig:comb_net}, we use
the terms receiver and $r$-subset interchangeably.

We use these notations throughout the paper:

\begin{itemize}
\item[-] $[m] = \{ 1,\ldots, m \} $ and $[m] _\ast = \{0,1,\ldots,m\}
$

\item[-] $A^{r}=\{ ( s_{1},\ldots ,s_{r}) :s_{i}\in
A\} $. Note that there is no limit on the number of times an element
$s_k$ in set $A$ can be chosen in $(s_1,\ldots,s_r)$.

\item[-] $A^{[r]}=\{ [ s_{1},\ldots ,s_{r}]
:s_{i}\in A\text{ and }s_{i}\neq s_{j}\text{ for }i\neq j\} $. In other words,
$A^{[r]}$ is the set of ordered vectors with distinct elements.

\item[-] $\mathbf{d} _{<F}^{u}(\cdot)$ is an operator on polynomials which
truncates to the terms of degree less than $F$ with respect to $u$.
\end{itemize}
Furthermore, we use the notation $\mathbb{I}$ for indicator function, defined as
\[ \mathbb{I}_{(\omega\in\Omega)} = \left\{
 \begin{array}{rll}
  1 &  \text{if} & \omega\in\Omega\\
  0 &  \text{if} & \omega\notin\Omega
 \end{array} \right.
.\]

For an allocation $\left( x_{1},\ldots ,x_{n}\right) $ of $T$ symbols, let $\Psi
\left( x_{1},\ldots ,x_{n}\right)$ count the number of unsuccessful receivers. We
can write $\Psi$ as
\begin{equation*}
\Psi \left( x_{1},\ldots ,x_{n}\right) :=\sum_{\substack{ S\subseteq [n]  \\ |S| =r}}\mathbb{I}\left( \tsum_{i\in
S}x_{i}<F\right)
,\end{equation*}
where the first sum is over all the subsets of size $r$ of storage nodes.
Therefore, the allocation problem we consider is the following optimization problem:
\begin{equation*}
\begin{array}{ll}
\text{minimize	}			& \Psi \left( x_{1},\ldots ,x_{n}\right) \\
\text{subject to}			& \sum_{i=1}^n x_{i} = T \\
							& x_i \geq 0, \quad x_i \text{ integer}
\end{array}
.\end{equation*}
It is challenging to find the optimal allocation because of the large space of
possible allocations, non-convexity, and discontinuity of the indicator
function. Our approach for solving this problem is to look into another
quantity which, for $r \ll \sqrt{n}$, closely approximates
$\Psi$ but it is easier to compute.
\subsection{Main Result}
\label{subsec:result}

Let $\alpha _{k}$ be the fraction of nodes containing $k$ symbols, and let
$c:=T/n$. The set of constraints on admissible allocation with respect to
$\alpha$ can be re-written as \[\left\{ \alpha :\tsum_{i=0}^{F}\alpha _{i}=1
\text{ and }\tsum_{i=0}^{F}i\alpha _{i} = c\text{ }\right\} .\] Given an
allocation $(x_1,\ldots,x_n)$, we can compute the parameters
$\alpha_0,\ldots,\alpha_F$. Then, we define $\varphi ( \alpha _{0},\ldots ,\alpha
_{F})$ as the probability that a receiver with access to a uniformly
chosen subset of nodes $\boldsymbol{s}$ from $[n]^r$ (shown by
$\boldsymbol{s}\sim [n]^r$) is unsuccessful in recovering the file. We have
\begin{equation}
\label{eq:phi_expression}
\varphi ( \alpha _{0},\ldots ,\alpha _{F})
:=\mathbf{P}_{\mathbf{s\sim }[n] ^{r}}\left( \dsum_{i=1}^{r}x_{
\boldsymbol{s}_{i}}<F \mid \alpha \right)
.\end{equation}
Our first claim says that $\varphi$ is a good approximation for $\Psi$.
\begin{theorem}
\label{thm:mainresult}
\begin{equation*}
\inf_{\alpha}\varphi ( \alpha) -\frac{2(r-1)^2}{n}\leq
\frac{r!}{n^{r}} \inf_{x\in \text{Allocations}}\Psi ( x) \leq
\inf_{\alpha}\varphi ( \alpha )
.\end{equation*}
\end{theorem}
\begin{proof}
Proof is given in Sec.~\ref{sec:proofresult}.
\end{proof}
Dealing with the functional $\varphi$ is simpler than working with $\Psi$. In the
definition of $\varphi$, the random vector is chosen from $[n]^r$ where
repetition is allowed. As a result, the probability generating function of
$\varphi$ has a simple form and is easy to work with. The main result of the
Theorem~\ref{thm:mainresult} is that for $r \ll
\sqrt{n}$, we can solve the problem of minimizing $\varphi( \alpha _{0},\ldots
,\alpha_{F})$ instead, which is simpler than solving for the original
optimization problem. Moreover, this solution is also a good approximation of
the problem. We will further discuss the discrepancy in the optimal solution
through an example in the last section.

From this point on, we will drop the conditioning on $\alpha$ for
brevity. Please note that $\varphi$ is just a function of $ ( \alpha _{0},\ldots
,\alpha _{F})$ and its value remains the same for all allocations with the same $
( \alpha _{0},\ldots ,\alpha _{F})$.

\section{Discussions and Proof of the Main Result}
\label{sec:proofresult}
Consider a receiver which has access to the vector of storage nodes
$\boldsymbol{s}=[ \boldsymbol{s}_1,\ldots ,\boldsymbol{s} _r] $, where
$\boldsymbol{s}$ is uniformly chosen from $[n]^{[r]}$. Let $\mathbf{P}_{\mathbf{
s\sim }[n] ^{[r] }}\bigl[ \sum_{i=1}^{r}x_{ \boldsymbol{s}_{i}}<F\bigr]$
represent the probability that the total number of symbols stored in a randomly
chosen set of size $r$ of storage nodes $\boldsymbol{s}$ is less than the file
size $F$. There are in total $n(n-1)\ldots(n-(r-1))=r!\binom{n}{r}$ ordered
vectors like $\boldsymbol{s}$ in $[n]^{[r]}$ (note that the subsets in
$[n]^{[r]}$ are ordered). Although working with ordered sets is slightly more
complicated, as we will see shortly, this will help us in finding a better
approximation for $\Psi$.

The total number of unsuccessful receivers $\Psi \bigl(x_{1},\ldots ,x_{n}\bigr)
$ can be calculated easily if we have the probability
$\mathbf{P}_{\mathbf{s}\sim[n]^{r}}$ that a receiver with access to a randomly
chosen subset of nodes $\boldsymbol{s}$ is unsuccessful in recovering the file.
In the definition of $\Psi$, we are only concerned with the total number of
unsuccessful receivers. If we choose $\boldsymbol{s}$ from a space like
$[n]^{[r]}$ where order is important, we need to eliminate the effect of
over-counting. Here, since $\boldsymbol{s}\sim[n]^{[r]}$, a division by $r!$ is
sufficient. Hence, the functional $\Psi \bigl(x_{1},\ldots ,x_{n}\bigr) $ can be
re-written as
\begin{equation}
\label{eq:prob-expression}
\Psi \bigl( x_{1},\ldots ,x_{n}\bigr) =\binom{n}{r}\mathbf{P}_{\mathbf{
s\sim }[n]^{[r]}}\bigl[ \tsum_{i=1}^{r}x_{
\boldsymbol{s}_{i}}<F\bigr]
.\end{equation}
In order to prove Theorem~\ref{thm:mainresult}, we first derive the lower
bound on $\varphi$ and then we prove the upper bound in the lemmas below.

\newtheorem{lem1}{Lemma}
\begin{lem1}
\label{lem:lower bound}
For any allocation $(x_1,\ldots , x_n)$, satisfying a given set of $\alpha$'s,
the following hold:
\begin{equation}
\label{eq:2}
 \frac{r!\Psi \bigl( x_{1},\ldots ,x_{n}\bigr) }{n^{r}}  \leq  \varphi
\bigl( \alpha _{0},\ldots ,\alpha _{F}\bigr)
.\end{equation}
\end{lem1}
\begin{proof}[Proof of Lemma~\ref{lem:lower bound}]
The inequality \eqref{eq:2} follows immediately from the definitions of $\varphi$
and $\Psi$ (\eqref{eq:phi_expression} and \eqref{eq:prob-expression}).
\end{proof}

In order to prove the upper bound, we need to look at the total variation
between the distributions of a uniform random vector $\mathbf{s\sim }[n]
^{r}$ and a uniform random vector $\mathbf{s}^{\prime }\mathbf{\sim }[n] ^{[r] }$.
\begin{definition} The total variation of two probability distributions $\mu $ and $\nu $ on a
discrete space $\Omega $ is defined as
\begin{equation*}
\limfunc{TV}( \mu ,\nu ) =\sup_{A\subseteq \Omega }\bigl\vert \mu
\bigl( A\bigr) -\nu \bigl( A\bigr) \bigr\vert \text{.}
\end{equation*}
\end{definition}
A well known integral formula for the total variation between two distributions
is given by \[\limfunc{TV}\bigl( \mu ,\nu \bigr)
=\frac{1}{2}\tsum\limits_{\omega \in \Omega }\bigl\vert \mu \bigl( \omega \bigr) -\nu \bigl( \omega \bigr)
\bigr\vert.\]
\newtheorem{lemTV}[lem1]{Lemma}
\begin{lemTV}
\label{lem:TV}
Let $\Omega=[ n] ^{r}$. Further, let $\mu $ be the uniform
probability distribution over $\Omega $, and $\nu $ be the uniform
probability distribution over the subset $S$ of $\Omega $ consisting of
vectors with distinct entries; $\nu$ is $0$ on $\Omega \setminus S$. Then, we
have \[ \limfunc{TV}(\mu,\nu) \leq \frac{(r-1)^2}{n}\text{.} \]
\end{lemTV}
\begin{proof}[Proof of Lemma~\ref{lem:TV}]
The total number of non-repetitive vectors of size $r$ in $\Omega$ is
$n(n-1)\ldots(n-r+1)$. We use the short hand $n^{[r]}$ for this expression. Then
we can write
\begin{align*}
\limfunc{TV}( \mu ,\nu )  =&\frac{1}{2}\Bigl\{
\tsum_{\omega \in S}\vert \mu ( \omega) -\nu (
\omega ) \vert
+ \Bigl. \tsum_{\omega \notin S}\vert \mu( \omega ) -\nu( \omega ) \vert \Bigr\}  \\
=&\frac{1}{2}\Bigl\{ \tsum_{\omega \in S}\Bigl( \frac{1}{n^{[ r] }}-\frac{1}{n^{r}}\Bigr) +\tsum_{\omega \notin S}\frac{1}{
n^{r}}\Bigr\}   \\
=& 1-\frac{n^{[r] }}{n^{r}}  = 1- \prod_{l=0}^{r-1}\Bigl(1-\frac{l}{n}\Bigr)
< 1-\Bigl(1-\frac{r-1}{n}\Bigr)^{r-1} \\
<&1-\Bigl[1-(r-1).\frac{r-1}{n}\Bigr]
=\frac{(r-1)^2}{n}.
\end{align*}
\end{proof}
\newtheorem{lemUpper}[lem1]{Lemma}
\begin{lemUpper}
\label{lem:upper bound}
\begin{equation}
\label{eq:3}
 \frac{r!\Psi \bigl( x_{1},\ldots ,x_{n}\bigr) }{n^{r}}  \geq \varphi
\bigl( \alpha _{0},\ldots ,\alpha _{F}\bigr) -\frac{2(r-1)^2}{n}
\end{equation}
\end{lemUpper}
\begin{proof}[Proof of Lemma~\ref{lem:upper bound}]
Using the results of Lemma~\ref{lem:TV} and definitions of $\Psi$
and $\varphi$, we can write
\begin{align*}
&\Bigl\vert \frac{r!\Psi \Bigl( x_{1},\ldots ,x_{n}\Bigr) }{n^{r}}-\varphi
\Bigl( \alpha _{0},\ldots ,\alpha _{F}\Bigr) \Bigr\vert  \\
=&\Bigl\vert \frac{n^{[ r] }}{n^{r}}\limfunc{P}\nolimits_{\mathbf{s\sim }
[ n] ^{[ r] }}\Bigl[ \tsum_{i=1}^{r}x_{\boldsymbol{s}
_{i}}<F\Bigr] -\limfunc{P}\nolimits_{\mathbf{s\sim }[ n]
^{r}}\Bigl( \tsum_{i=1}^{r}x_{\boldsymbol{s}_{i}}<F\Bigr) \Bigr\vert  \\
\leq &\Bigl\vert \mathbf{P}_{\mathbf{s\sim }[n] ^{[r] }}
\Bigl[ \tsum_{i=1}^{r}x_{\boldsymbol{s}_{i}}<F\Bigr] -\mathbf{P}
_{\mathbf{s\sim }[n] ^{r}}\Bigl( \tsum_{i=1}^{r}x_{
\boldsymbol{s}_{i}}<F\Bigr) \Bigr\vert  \\
&\quad +\Bigl( 1-\frac{n^{[r] }}{n^{r}}\Bigr)
\leq \frac{2(r-1)^2}{n}
.\end{align*}

The first inequality above follows from the triangle inequality, and the second
from Lemma~\ref{lem:TV}.

\end{proof}

The proof of the Theorem~\ref{thm:mainresult} follows from Lemma~\ref{lem:lower bound}
and~\ref{lem:upper bound}.
%
\section{Optimal Symmetric Allocations}
\label{sec:optimalsymmetric}

Following the results of the previous section, for cases where $r \ll \sqrt{n}$,
we have $\frac{2(r-1)^2}{n}\ll 1$ and therefore, finding the optimal allocation of
the symbols is equivalent to minimizing the function $\varphi ( \alpha
_{0},\ldots ,\alpha _{F})$. In this section, we direct our attention to symmetric
allocations. In the case of symmetric allocations, we can find the optimal
symmetric allocation and probability of success for all different budgets $T$. An
allocation is called symmetric if we allocate the budget $T$ as follows: we pick a
number, say $j$, and we allocate chunks of size $T/j$ until we run out of the
budget. Now, we have two types of nodes: fraction $\alpha_0$ of nodes which are
left empty and the fraction $\alpha_j$ of the nodes which store $j$ number of
symbols.

Again, the optimal allocation is not obvious even if we consider only symmetric
allocations. For instance, for very low budgets ($T\approx F$), we can easily
argue that the budget should be concentrated over a minimal subset of nodes. For
example, consider the case where $T=F$, if we store the entire file over one of
the storage nodes, then the total number of successful receivers is
$\binom{n-1}{r-1}$. If we break the file into two parts each of size $F/2$, then
the total number of successful receivers is going to be $\binom{n-2}{r-2}$. By
using the well-known identity \[\binom{n}{r}=\binom{n-1}{r-1}+\binom{n-1}{r},\]
it is clear that the former allocation outperforms the latter. Similarly, other
symmetric allocations can also be rejected. When the budget is very high ($T
\approx nF/r$), the budget should be spread maximally. For example, consider the
case where $T = nF/r$. In this case, by spreading the budget over all the storage
nodes, we can achieve the probability-1 recovery. If one distributes this budget
by allocating chunks of size $F$ (storing the file in its entirely), he will be
worse-off since the probability of success will be \[1-\binom{n-\lfloor
n/r\rfloor}{r},\] which is clearly less than 1. This behavior gives rise to
questions like: ``When to switch from minimal to maximal spread of the budget?'',
``Is there any situation where there exists a solution other than minimal or
maximal spreading?''

First, we give a useful expression for $\varphi$ in the lemma below, which is
simpler to work with. Then, we investigate the optimal symmetric allocation.
\newtheorem{lem4}[lem1]{Lemma}
\begin{lem4}
\label{lem:phigeneratingfunction}
\begin{equation}
\label{eq:1}
\varphi \left( \alpha_{0},\ldots ,\alpha _{F}\right) = \left[
\mathbf{d} _{<F}^{u}\left( \sum_{k=0}^{F}u^{k}\alpha _{k}\right) ^{r}\right]
_{_{\upharpoonright u=1}}
\end{equation}
\end{lem4}

\begin{proof}[Proof of Lemma~\ref{lem:phigeneratingfunction}]
If $ \boldsymbol{s}_{i}$ is a random element of $[n] $, then the probability that
$\mathbf{P}(x_{\boldsymbol{s}_{i}} = k)$ is equal to $\alpha_k$. Therefore, the
probability generating function of $x_{\boldsymbol{s}_{i}}$ is equal to $
\sum_{k=0}^{F}u^{k}\alpha _{k}$. Hence, if $\boldsymbol{s}=\bigl(
\boldsymbol{s}_{1},\ldots ,\boldsymbol{s}_{r}\bigr) $ is a uniform random vector
in $[n] ^{r}$, then the probability generating function of
$\sum_{i=1}^{r}x_{\boldsymbol{s}_{i}}$ is equal to $\bigl(
\sum_{k=0}^{F}u^{k}\alpha _{k}\bigr) ^{r}$. It follows then that
\begin{equation}
\mathbf{P}_{\mathbf{s\sim }[n] ^{r}}\bigl(
\tsum_{i=1}^{r}x_{\boldsymbol{s}_{i}}<F\bigr) =\left[
\mathbf{d} _{<F}^{u}\bigl( \tsum_{k=0}^{F}u^{k}\alpha _{k}\bigr) ^{r}\right]
_{_{\upharpoonright u=1}} .
\end{equation}
and~\eqref{eq:1} is immediate.
\end{proof}

In a symmetric allocation, suppose that the fraction of the non-empty nodes is
$\alpha_j$ with $j$ number of symbols each. Therefore, in the
expression of $\varphi( \alpha _{0},\ldots ,\alpha_{F})$ at most $\alpha_0$ and
$\alpha_j$ have non-zero values. Our goal is to find the optimal value of $j$.

In this case, using Lemma~\ref{lem:phigeneratingfunction}, the problem of
minimizing $\varphi ( \alpha _{0},\ldots ,\alpha _{F})$ over $\{\alpha_0+\alpha_j
= 1, j\alpha_j = c \}$ reduces to
\begin{equation}
\label{eq:symmetricExpansion}
 \varphi (\alpha_j) = \left[ \mathbf{d} _{<F}^{u} (\alpha_0 + u^j \alpha _j)^{r}
 \right] _{_{\upharpoonright u=1}}
.\end{equation}
Equivalently, by substituting $(1-\alpha_j)$ for $\alpha_0$, we have
\begin{equation}
\label{eq:binomial}
\varphi ( \alpha _j) = \sum^{\lfloor (F-1)/j \rfloor}_{i=0}
\binom{r}{i}\alpha_j^i (1-\alpha_j)^{r-i} \quad \text{for} \quad  rj \geq F.
\end{equation}
Notice that for $rj < F$, the maximum degree of $u$
in~\eqref{eq:symmetricExpansion} is less than $F$. Therefore, the operator
$\mathbf{d}_{<F}^u$ does not eliminate any term from the expansion and
$\varphi(\alpha_j)=1$.

Expression~\eqref{eq:binomial} has the form of the binomial distribution CDF; The
following lemma helps us to determine its minima.
\newtheorem{lem5}[lem1]{Lemma}
\begin{lem5}
\label{lem:local min}
The function $\varphi(\alpha_j)$ in~\eqref{eq:binomial} has a local minimum in
all the points $j$ where $\lfloor \frac{F-1}{j-1} \rfloor-\lfloor
\frac{F-1}{j} \rfloor \geq1$. In other words, $\varphi(\alpha_j)$ minimizes
over some $j^\star$ of the form $\lceil\frac{F-1}{i} \rceil$ for some $i$.
\end{lem5}
\begin{proof}
For constants $m$ and $n$,
$f(x)=\sum_{i=0}^{m}\binom{n}{i}x^{i}\left(1-x\right)^{n-i}$ is decreasing in
$x$. Therefore, if $j_{1}<j_{2}$ and $\lfloor \frac{F-1 }{j_{1}}\rfloor = \lfloor \frac{F-1}{j_{2}}\rfloor $, then $ \alpha
_{j_{1}}>\alpha _{j_{2}}$ and thus, $\varphi( \alpha _{j_{1}})
<\varphi( \alpha _{j_2}) $.
\end{proof}

Lemma~\ref{lem:local min} reduces the complexity of finding the minimum
of~\eqref{eq:binomial} considerably, as it limits the search for optimal $j$,
shown by $j^\star$, to the set of per node budgets $\{\lceil \frac{F}{i}
\rceil:i\in[r]\}$. Therefore, finding the optimal symmetric allocation is reduced
to computing the probability of successful recovery of the original file when
$\alpha_{j^\star}$ fraction of the nodes contain $j^\star$ portion of the file
and the rest of the nodes are empty.

In order to find the optimal value $j^\star$, we derive the probability of
successful decoding of a random receiver. Suppose that only $d$ out of $r$ of
storage nodes to which a receiver has access are non-empty. In this case, the
receiver can recover the file only if $d\ge i$. Therefore, the probability of
successful file recovery when each non-empty storage node has $\lceil
\frac{F}{i} \rceil$ portion of the file is
\begin{equation}
\label{eq:hypergeo}
\frac{1}{\binom{n}{r}}\sum_{d=i}^r \binom{\frac{T}{\lceil \frac{F}{i}
\rceil}}{d}
\binom{n-\frac{T}{\lceil \frac{F}{i}
\rceil}}{r-d},
\end{equation}
which has the from of the CDF of hyper-geometric distribution. We have to
evaluate this function for all $i \in [r]$ and choose $j^\star$ such that the
highest success probability is achieved. Note that given $r$ the solution can be
found in constant time since by Lemma~\ref{lem:local min} we just need to
evaluate~\eqref{eq:hypergeo} $r$ times.
\subsection{Symmetric Allocation in Connected Random Graphs}
\label{subsec:randomgraph}

In a practical network, a node cannot connect (via single hop) to every subset of
$r$ nodes. As a first step towards practical settings, we investigate the
asymptotics of the allocation problem in large random graphs. A random graph
$G(n,p)$ has $n$ vertices, and every two vertices are connected with probability $p$. 
We direct our attention to connected random graphs
since they better describe real networks. $G(n,p)$ is connected iff $p$ is
greater than a critical value $\frac{\log n}{n}$. If $p=\frac{d\log n}{n}$ for
some constant $d$, then $G\bigl(n,\frac{d\log n}{n}\bigr)$ is connected with high
probability and every vertex has degree $r\asymp \log n$~\cite{alon2008}.

Suppose that we want to store a file of size $F$ and budget $T$ in such a graph
provided that each node could reconstruct the file by accessing its 1-hop
neighbors. We are interested in maximizing the probability that a node is
successful, as the number of nodes $n$ in the network grows. It is clear that the
budget $T$ should also grow in order to maintain a certain success probability
for receivers. Otherwise, probability of successful recovery of the file will be
$0$. Given $T$, the mean number of symbols per node is $T/n$ and therefore the
mean number of symbols a node has access to is equal to $\frac{rT}{n}$. Since the
file size is assumed to be constant, the most important regime to study is when
$\frac{rT}{n}\asymp \mu $, where $\mu $ is a constant.

In this regime, every one of the random variables $x_{\mathbf{s}_{1}},\dots
,x_{\mathbf{s}_{r}}$, representing the number of symbols in every chosen node, is
a non-negative random variable with the expectation $\mu/r$. Standard limit
theorems (\cite{petrov1995limit}) imply that the random variable
$\sum_{i=1}^{r}x_{ \boldsymbol{s}_{i}}$ will follow approximately a Poisson
distribution. Consider the case $r\asymp d\log n$ and $T=\mu n/r$. For
$i=1,\ldots ,F$, define $\lambda _{i}$ so that $\alpha _{i}=\frac{\lambda
_{i}}{r}$ and let $ X_{1},\ldots, X_{F}$ be independent Poisson random variables
such that $X_{i}$ follows $\limfunc{Poisson}( k;\lambda _{i}) = \lambda_i^k
e^{\lambda_i}/k!$. Then, classic approximation theorems~(\cite
{le1960approximation, barbour1988stein}) imply that the random variables
$\sum_{i=1}^{r}x_{\boldsymbol{s}_{i}}$ and $\sum_{i=1}^{F}iX_{i}$ behave
similarly. In fact, their difference in total variation obeys the following
bound
\begin{equation*}
\limfunc{TV}\Bigl( \tsum\limits_{i=1}^{r}x_{\boldsymbol{s}
_{i}},\tsum\limits_{i=1}^{F}iX_{i}\Bigl) =\limfunc{O}\Bigl(
\frac{1}{\log n}\Bigl)
.\end{equation*}
Therefore, it is the case that
\begin{equation*}
\mathbf{P}_{\mathbf{s\sim }[n] ^{r}}\Bigl(
\sum_{i=1}^{r}x_{\boldsymbol{s}_{i}}<F\Bigl) =\mathbf{P}\Bigl(
\sum\limits_{i=1}^{F}iX_{i}<F\Bigl) +\limfunc{O}\Bigl( \frac{1}{\log n} \Bigl)
.\end{equation*}

In the symmetric case, we allocate either $0$ or $j$ symbols. Hence, at most
$\lambda _{0}$ and $\lambda _{j}$ have non-zero values. Since in symmetric case
we have $j\alpha _{j}=r\mu $, the previous expression becomes
\begin{equation*}
\mathbf{P}_{\mathbf{s\sim }[n]^{r}}\Bigl(\dsum_{i=1}^{r}x_{\boldsymbol{s}%
_{i}}<F\Bigl)=\!\!\tsum\limits_{k=0}^{\bigl\lfloor\frac{F-1}{j}\bigr\rfloor}%
\frac{(\mu /j)^{k}e^{-\mu /j}}{k!}+\limfunc{O}\Bigl(\frac{1}{\log n}\Bigr)%
\text{.}
\end{equation*}
Similar to the result in the previous section, since $e^{-x}%
\sum_{k=0}^{m}x^{k}/k!$ is a decreasing in $x$, in order to find the optimal 
$j$, we just need to evaluate the above expression for $j\in \{\lceil
F/i\rceil :i\in [r]\}$ and the optimal value $j^{\star }$ is the one
which maximizes the success probability. 
\section{Simulation Results and Conclusion}
\label{sec:simulation}

We numerically investigated the results of section~\ref{sec:optimalsymmetric}
through some simulations. Due to the complexity of the problem, finding the true
optimal allocation for large $n$ is not practical. In order to verify our
results, we compare the approximate solution with optimal (found by searching all
symmetric allocations) for two different examples. First, for $n=10$ and $r=2$,
optimal symmetric allocation consists of two parts: for $T/F \in (1,4.5)$, the
file should be stored entirely and, for $T/F > 4.5$, all storage locations should
store half of the file. As shown in Figure~\ref{fig:j}, approximate solution
gives correct allocation for this case. For the second case, where $n=15$
and $r=5$, the optimal allocation is more complicated. We observe that the choice of
$j/F=1$ remains optimal until $T/F=4.5$. Then, for $T/F \in (4.5,4.65)$, the
optimal number of nodes to use is 9 ($=\lfloor T/j^\star \rfloor$) and each of
them store half of the file. Finally, we observe a transition that spreads the
file maximally over all storage nodes. It is interesting that in this case our
approximate solution again matches the optimal symmetric allocation.
Figure~\ref{fig:prob} plots the probability of success versus normalized budget
for $n = 15$ and $r = 3$.
\begin{figure}[htp]
\begin{center}
    \vspace{-.2in}
  \includegraphics[width=2.4in]{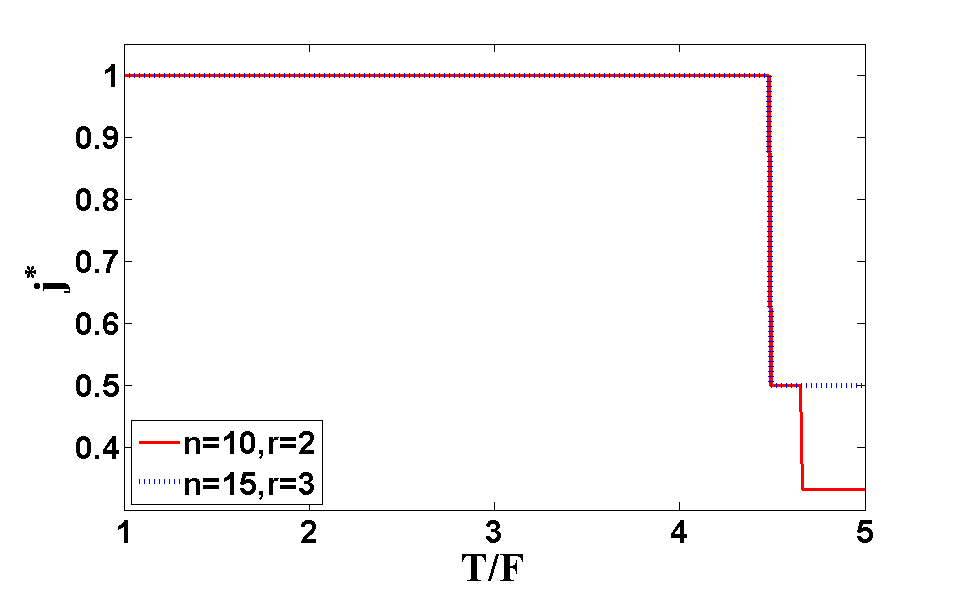}
  \vspace{-.15in}
  \caption{Optimal symmetric allocation vs normalized budget. }
  \label{fig:j}
  \vspace{-.15in}
\end{center}
\end{figure}
\begin{figure}[htp]
\begin{center}
  \includegraphics[width=2.3in]{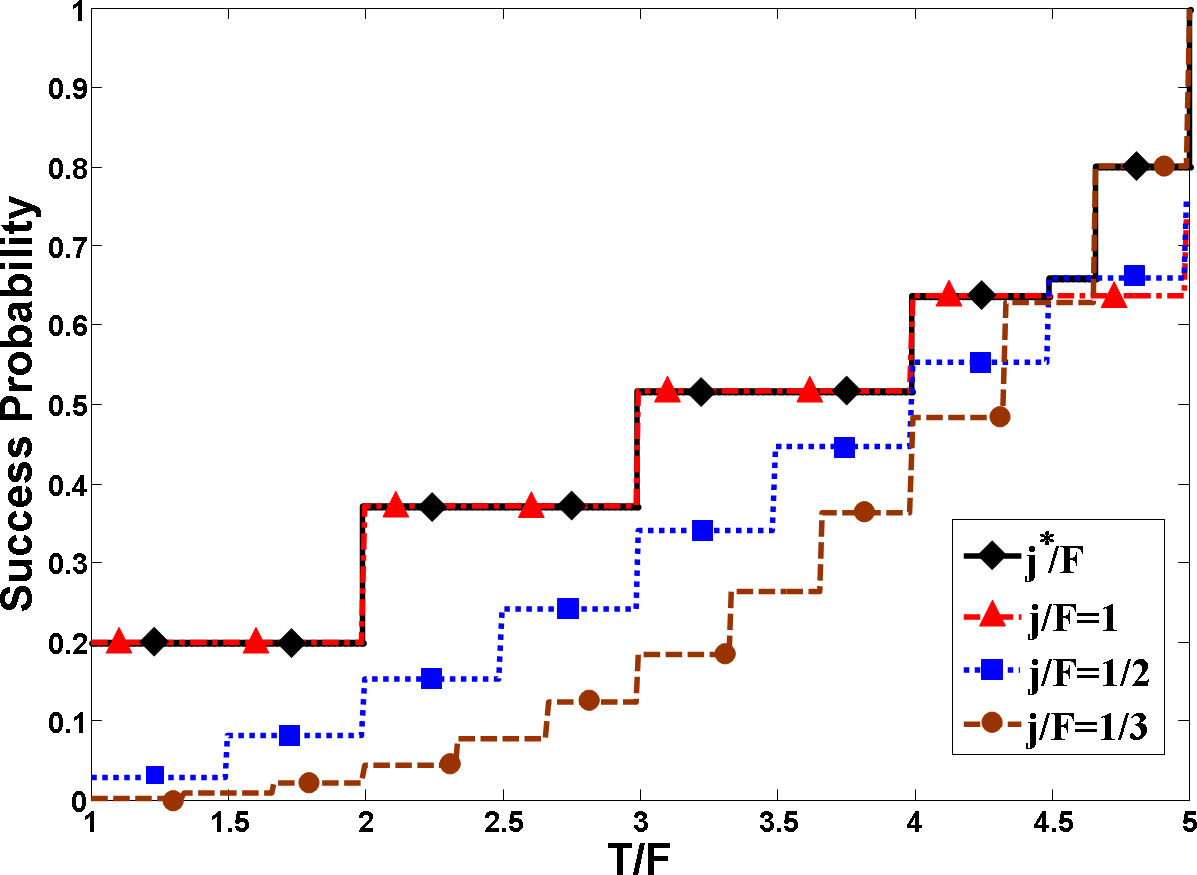}
  \vspace{-.1in}
  \caption{Probability of success vs normalized budget for $n=15$ and $r=3$ in the
  symmetric allocation where each node stores $j/F$ fraction of the file.}
  \label{fig:prob}
  \vspace{-.25in}
\end{center}
\end{figure}

In general, we observe a transition from concentration of budget over minimal
number of nodes to maximal spreading of the budget over all storage nodes as the
budget increases (this observation is also reported in~\cite{Leong2009a}). This
transition is not sharp as we observed that there are cases where the number of
non-empty nodes is neither of the extremes. Also, where the transition happens is
not trivial to determine and for each budget the optimal allocation should be
computed using the machinery developed in this paper. Finding useful algorithms
in order to find the optimal allocation in general sense and also for more
realistic scenarios remains of interest.
\bibliographystyle{IEEEtran}
\bibliography{myBib}

\begin{thebibliography}{10}
\providecommand{\url}[1]{#1}
\csname url@samestyle\endcsname
\providecommand{\newblock}{\relax}
\providecommand{\bibinfo}[2]{#2}
\providecommand{\BIBentrySTDinterwordspacing}{\spaceskip=0pt\relax}
\providecommand{\BIBentryALTinterwordstretchfactor}{4}
\providecommand{\BIBentryALTinterwordspacing}{\spaceskip=\fontdimen2\font plus
\BIBentryALTinterwordstretchfactor\fontdimen3\font minus
  \fontdimen4\font\relax}
\providecommand{\BIBforeignlanguage}[2]{{%
\expandafter\ifx\csname l@#1\endcsname\relax
\typeout{** WARNING: IEEEtran.bst: No hyphenation pattern has been}%
\typeout{** loaded for the language `#1'. Using the pattern for}%
\typeout{** the default language instead.}%
\else
\language=\csname l@#1\endcsname
\fi
#2}}
\providecommand{\BIBdecl}{\relax}
\BIBdecl

\bibitem{Naor1991}
M.~Naor and R.~M. Roth, ``Optimal file sharing in distributed networks,'' in
  \emph{Proc. nd Annual Symposium on Foundations of Computer Science}, 1991,
  pp. 515--525.

\bibitem{Jiang2003}
A.~Jiang and J.~Bruck, ``Memory allocation in information storage networks,''
  in \emph{Proc. IEEE International Symposium on Information Theory}, 2003, pp.
  453--.

\bibitem{Jiang2005}
------, ``Network file storage with graceful performance degradation,''
  \emph{Trans. Storage}, vol.~1, no.~2, pp. 171--189, 2005.

\bibitem{Zhenning2009}
Z.~Kong, S.~A. Aly, and E.~Soljanin, ``Decentralized coding algorithms for
  distributed storage in wireless sensor networks,'' \emph{CoRR}, vol.
  abs/0904.4057, 2009.

\bibitem{Dimakis2005}
A.~G. Dimakis, V.~Prabhakaran, and K.~Ramchandran, ``Ubiquitous access to
  distributed data in large-scale sensor networks through decentralized erasure
  codes,'' in \emph{Proc. Fourth International Symposium on Information
  Processing in Sensor Networks IPSN 2005}, 2005.

\bibitem{Leong2009a}
D.~Leong, A.~G. Dimakis, and T.~Ho, ``Distributed storage allocation
  problems,'' in \emph{Proc. Workshop on Network Coding, Theory, and
  Applications NetCod '09}, 2009, pp. 86--91.

\bibitem{alon2008}
N.~Alon and J.~Spencer, \emph{The Probabilistic Method--3rd edition}.\hskip 1em
  plus 0.5em minus 0.4em\relax John Wiley \& Sons, USA, 2008.

\bibitem{petrov1995limit}
V.~Petrov, \emph{{Limit theorems of probability theory: sequences of
  independent random variables}}.\hskip 1em plus 0.5em minus 0.4em\relax Oxford
  University Press, USA, 1995.

\bibitem{le1960approximation}
L.~Le~Cam, ``{An approximation theorem for the Poisson binomial
  distribution},'' \emph{Pacific J. Math}, vol.~10, no.~4, pp. 1181--1197,
  1960.

\bibitem{barbour1988stein}
A.~Barbour, ``{Stein's method and Poisson process convergence},'' \emph{Journal
  of Applied Probability}, pp. 175--184, 1988.

\end{thebibliography}

\end{document}